\documentclass[letterpaper, 10pt, onecolumn,conference]{ieeeconf}  

\IEEEoverridecommandlockouts                              

\usepackage{graphicx} 
\usepackage{mathptmx} 
\usepackage{times} 
\usepackage{amsmath} 
\usepackage{amssymb}  
\usepackage{datetime}
\usepackage{hyperref}
\usepackage{color}

\newtheorem{thm}{Theorem}

\def\rank{\mathrm{rank}}
\DeclareMathAlphabet{\bit}{OML}{cmm}{b}{it}

\def\<{\leqslant}           
\def\>{\geqslant}           

\def\d{\partial}
\def\wh{\widehat}
\def\wt{\widetilde}

\def\Re{\mathrm{Re} }   
\def\Im{\mathrm{Im} }   

\def\mR{{\mathbb R}}
\def\mC{\mathbb{C}}

\def\Tr{\mathrm{Tr}}

\def\rT{\mathrm{T}}

\def\bS{\mathbf{S}}

\def\bE{\mathbf{E}}

\def\[[[{[\![\![}
\def\]]]{]\!]\!]}

\def\bra{\langle}
\def\ket{\rangle}

\def\re{\mathrm{e}}
\def\rd{\mathrm{d}}

\def\bJ{\mathbf{J}}

\def\x{\times}
\def\ox{\otimes}

\def\fF{\mathfrak{F}}

\def\fH{\mathfrak{H}}
\def\fS{\mathfrak{S}}

\def\mS{{\mathbb S}}

\def\eps{\epsilon}
\def\ups{\upsilon}
\def\Ups{\Upsilon}

\def\diag{\mathop\mathrm{diag}}

\def\od{\odot}

\title{\LARGE \bf
Decoherence Time in Quantum Harmonic Oscillators as Quantum Memory Systems*}

\author{Igor G. Vladimirov$^{1}$, \qquad Ian R. Petersen$^{2}$
\thanks{*This work is supported by the Australian Research Council  grants  DP210101938, DP200102945.}
\thanks{$^{1,2}$School of Engineering, Australian National University, Australia, Acton, ACT 2601, Canberra, Australia
        {\tt\small igor.g.vladimirov@gmail.com, i.r.petersen@gmail.com}.}
}

\begin{document}
\maketitle
\thispagestyle{empty}
\pagestyle{empty}

\begin{abstract}
This paper is concerned with open quantum harmonic oscillators (OQHOs)  described by linear quantum stochastic differential equations. This framework includes isolated oscillators with zero Hamiltonian, whose system variables  remain unchanged (in the Heisenberg picture of quantum dynamics)  over the course of time, making such systems potentially applicable as quantum memory devices. In a more realistic case of system-environment coupling, we  define a  memory decoherence horizon  as a typical time for a mean-square deviation of the system variables from their initial values to become relatively  significant as specified by a weighting matrix and a fidelity parameter.  We consider the maximization of the decoherence time over the energy and coupling matrices of the OQHO as a memory system in its storage phase and obtain a condition under which the zero Hamiltonian delivers a suboptimal solution. This optimization problem is also discussed for an interconnection of OQHOs.
\end{abstract}

\section{INTRODUCTION}

As a modelling paradigm for the physical world at atomic and subatomic scales, quantum mechanics \cite{LL_1981} has  essential distinctions from classical mechanics and classical field theory. They include the noncommutative operator-valued nature of quantum dynamic variables and quantum states (both represented by operators on an underlying Hilbert space) and more complicated probabilistic and information theoretic aspects of  quantum system dynamics subject to measurement \cite{H_2001}.  These distinctions and their consequences are exploited as quantum mechanical resources in quantum technologies, such as quantum communication and quantum information processing  \cite{NC_2000},  with the aim of engineering quantum systems to outperform their classical analogues.
The performance criteria used for this purpose depend on specific control objectives  pertaining to qualitative aspects and optimization of the system behaviour (for example, stability, robustness to unmodelled dynamics,  and minimization of cost functionals).

In particular, applications to information storage employ the ability of a system to retain certain states or dynamic variables over the course of time in combination with the possibility to ``write'' and ``read'' them (see \cite{FCHJ_2016,YJ_2014} and references therein).
For eliminating the external interference in the storage stage,  the system is isolated from the surroundings, and its internal dynamics are described in terms of a Hamiltonian, which is usually a function (for example, a polynomial) of the system variables.
As a result, any operator which commutes with the Hamiltonian (including the Hamiltonian itself) is preserved over the course of time in this case (we consider the Heisenberg picture of quantum dynamics).
Therefore, if the Hamiltonian of the isolated system is zero, then all the system variables are preserved in time, thus providing a set of integrals of motion in the form of (in general noncommuting) quantum variables.
However, such integrals of motion usually hold only approximately (and can be regarded as conserved quantities over a finite time horizon) because of model inaccuracies and the coupling of the system to its environment, including classical or quantum systems and external fields.

In the framework of the Hudson-Parthasarathy calculus \cite{HP_1984,P_1992},
the internal system dynamics and the system-field interaction (accompanied by the energy exchange between them) are governed by quantum stochastic differential equations (QSDEs) specified by the Hamiltonian and coupling operators (which are also functions of the system variables). In the presence of quantum noise, the system variables drift away from their initial conditions which would be conserved operators in the ideal case of isolated zero-Hamiltonian system dynamics. A typical time it takes for this deviation to become significant  in comparison with the initial values suggests a performance index for the quantum system as a memory device in its storage phase, while there also are other approaches to decoherence (see \cite{VP_2023_SCL} and references  therein).

The present paper discusses a particular way of quantifying the memory decoherence time for OQHOs,  whose dynamic variables are organized as conjugate pairs of positions and momenta 
\cite{S_1994} satisfying canonical commutation relations (CCRs)  and evolving in time according to linear QSDEs. Such systems provide an important class of models in quantum optics \cite{WM_2008} (for example, for modelling optical cavities) and, more generally,  are used as principal building  blocks in linear quantum systems theory \cite{NY_2017,P_2017,ZD_2022} which develops control and filtering approaches for them.
This class of systems (and their finite-level counterparts governed by quasilinear QSDEs \cite{EMPUJ_2016,VP_2022_SIAM}) lend themselves to computation of second and higher-order (for example, quadratic-exponential \cite{VPJ_2018a,V_2023_IFAC}) moments of the system variables, which is particularly efficient  in the case of external fields in the vacuum state \cite{P_1992}. We take advantage of this tractability and use a mean-square deviation of the system variables from their initial values, along with a weighting matrix and a fidelity parameter,  in order to define the decoherence time relevant to quantum memory applications. This leads to a problem of maximizing the decoherence time as a performance index over the energy and coupling matrices of the OQHO or its approximate version based on a second-order Taylor expansion of the decoherence time over the fidelity parameter.  We establish a condition under which the zero Hamiltonian provides a suboptimal solution and discuss this optimization problem for a coherent feedback interconnection of two OQHOs with direct energy and field-mediated coupling \cite{ZJ_2011a}.


\section{ISOLATED QUANTUM HARMONIC OSCILLATOR}
\label{sec:iso}

In the Heisenberg picture of quantum mechanics \cite{S_1994}, physical quantities are represented by time-varying operators $\zeta(t)$ on a Hilbert space $\fH$,  and a  fixed positive semi-definite self-adjoint   density operator (or quantum state)  $\rho = \rho^\dagger \succcurlyeq 0$ of unit trace $\Tr \rho = 1$ on $\fH$ specifies the quantum probabilistic structure including the expectation $\bE \zeta := \Tr (\rho \zeta)$. In this picture, a quantum harmonic oscillator is an autonomous dynamical system with an even number
\begin{equation}
\label{nnu}
  n:= 2\nu
\end{equation}
of self-adjoint operator-valued variables $X_1(t), \ldots, X_n(t)$ on the space $\fH$, which correspond to $\nu$ conjugate position-momentum pairs. Their evolution in time $t\> 0$ is governed by an ODE
\begin{equation}
\label{Xdot}
  \dot{X} =
  i[H,X]
  =
  A_0X
\end{equation}
for the column-vector $X:=(X_k)_{1\< k \< n}$,
where $\dot{(\ )}:= \rd/\rd t$ is the time derivative, and the matrix $A_0 \in \mR^{n\x n}$ is described below.  The superoperator $i[H, \cdot]$ uses the commutator $[\xi, \eta]:= \xi\eta - \eta\xi$ of linear operators (which applies entry-wise if $\eta$ is a vector of operators as in (\ref{Xdot})) and is a quantum mechanical counterpart of the Poisson bracket \cite{A_1989} for classical Hamiltonian systems (such as a multidimensional mass-spring system with $\nu$ degrees of freedom).  This superoperator is specified by a Hamiltonian
\begin{equation}
\label{H}
  H:= \tfrac{1}{2}X^\rT R X,
\end{equation}
which is a self-adjoint operator on $\fH$ representing the internal energy of the oscillator and depending on the system variables in a quadratic fashion
parameterized by an energy matrix $R = R^\rT \in \mR^{n\x n}$. Similarly to the Poisson brackets of the classical positions and momenta,  the quantum system variables, considered at the same but otherwise arbitrary moment of time,   satisfy the CCRs
\begin{equation}
\label{XCCR}
    [X,X^\rT]
    :=
    ([X_j,X_k])_{1\< j,k\< n}
    =
    2i\Theta,
\end{equation}
where $\Theta = -\Theta^\rT \in \mR^{n\x n}$ is a constant matrix. In particular, if the system variables are implemented as the multiplication operators $X_{2k-1} := q_k$ and the differential operators $X_{2k} := p_k:= -i\d_{q_k}$ for $k = 1, \ldots, \nu$ on the Schwartz space  \cite{V_2002},  
so that $[q_j,p_k] = i\delta_{jk}$ (with $\delta_{jk}$ the Kronecker delta), then
$\Theta = \frac{1}{2} I_\nu \ox \bJ$, where $I_\nu$ is the identity matrix of order $\nu$, $\ox$ is the Kronecker product, and
\begin{equation}
\label{bJ}
        \bJ
        : =
        {\begin{bmatrix}
        0 & 1\\
        -1 & 0
    \end{bmatrix}}.
\end{equation}
Irrespective of this special case (that is, for any real antisymmetric CCR matrix $\Theta$),
a combination of  (\ref{H}), (\ref{XCCR}) and the derivation property of the commutator leads to
\begin{equation}
\label{A0}
    A_0 = 2\Theta R
\end{equation}
in (\ref{Xdot}), which is a Hamiltonian matrix
in the sense of the symplectic structure specified by $\Theta$. Therefore, if $R=0$, that is, when  the oscillator has zero Hamiltonian, not only $H$ is preserved in time (which follows from $\dot{H} = i[H,H] = 0$ for any Hamiltonian, similarly to the classical case), but so also are the system variables $X_1, \ldots, X_n$ themselves.
The conservation of a set of noncommutative quantum  variables (rather than a single Hamiltonian) can be regarded as a quantum mechanical resource which makes such an oscillator potentially applicable as a storage device. We omit the discussion of ``initialization'' and ``retrieval''   of the system variables, which correspond   to the ``write'' and ``read'' memory operations and are beyond the scope of this paper. Instead, we will consider the effect of coupling between the quantum system and the environment (including the external fields), which makes the system variables evolve in time, usually in a dissipative fashion. Note that for the isolated oscillator (\ref{Xdot}) with a nonsingular CCR matrix, that is,
\begin{equation}
\label{detTheta}
  \det \Theta \ne 0
\end{equation}
(which is assumed in what follows) and a positive definite energy matrix $R\succ 0$, the matrix $A_0 = 2 R^{-1/2}\sqrt{R} \Theta \sqrt{R} \sqrt{R} $ in (\ref{A0}) is diagonalizable and has a purely imaginary spectrum (as any nonsingular real antisymmetric matrix). A similar argument applies to negative definite energy matrices $R$, in which case, $A_0$ is also diagonalizable with a purely imaginary spectrum. Therefore, if $R$ is positive or negative definite,  then the system variables perform an oscillatory motion with no dissipation because the dependence of $X(t) = \re^{tA_0} X(0)$ on the initial condition $X(0)$ does not fade away, as $t   \to +\infty$, and $X(t)$ does not ``lose memory'' about  $X(0)$.

\section{OPEN QUANTUM HARMONIC OSCILLATOR}
\label{sec:open}

In comparison with the isolated system dynamics (\ref{Xdot}), a more realistic setting is concerned with an OQHO in Fig.~\ref{fig:OQHO},
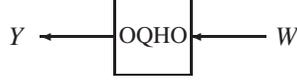
\begin{figure}[htbp]
\centering
\unitlength=1mm
\linethickness{0.5pt}
\begin{picture}(40.00,10.00)
    \put(15,-2){\framebox(10,10)[cc]{\small OQHO}}
    \put(35,3){\vector(-1,0){10}}
    \put(15,3){\vector(-1,0){10}}
    \put(38,3){\makebox(0,0)[cc]{$W$}}
    \put(2,3){\makebox(0,0)[cc]{$Y$}}
\end{picture}
\caption{A schematic depiction of an OQHO with the input field $W$ and output field $Y$ governed by (\ref{dXdY}).}
\label{fig:OQHO}
\end{figure}
whose internal and output variables evolve in time according to linear Hudson-Parthasarathy 
QSDEs \cite{NY_2017,P_2017}:
\begin{equation}
\label{dXdY}
    \rd X  = AX \rd t + B\rd W,
    \qquad
    \rd Y  = CX \rd t + D\rd W,
\end{equation}
where  $A \in \mR^{n\x n}$, $B \in \mR^{n\x m}$, $C \in \mR^{r\x n}$, $D\in \mR^{r\x m}$  are constant matrices specified below, with $m$, $r$ even. Here, $W:=(W_k)_{1\< k \< m}$ is a vector of quantum Wiener processes which are time-varying self-adjoint operators on a symmetric Fock space \cite{PS_1972} $\fF$ representing the input bosonic fields. Unlike the components of the  classical Brownian motion \cite{KS_1991} in $\mR^m$, these quantum processes do not commute with each other and  have a complex positive semi-definite Hermitian Ito matrix  $\Omega = \Omega^* \succcurlyeq 0$:
\begin{equation}
\label{Omega}
    \rd W \rd W^\rT
    = \Omega \rd t,
    \qquad
    \Omega: = I_m + iJ,
    \qquad
    J
    :=
    I_{m/2}\ox\bJ,
\end{equation}
with $(\cdot)^*:= \overline{(\cdot)}{}^\rT$ the complex conjugate transpose and
$\bJ$ from (\ref{bJ}).
Its imaginary part $J= \Im \Omega$ specifies the two-point CCRs for the quantum Wiener process  $W$ as
$  [W(s), W(t)^\rT]
  =
  2i \min(s,t)J
$ for all $ s, t \> 0$.
Also, $Y:= (Y_k)_{1\< k\< r}$ in (\ref{dXdY}) is a vector of $r$ time-varying self-adjoint operators on the space $\fH$,   which are selected from the  output fields of the OQHO resulting from its interaction with the input fields. Accordingly, $D$ consists of conjugate pairs of  $r\< m$ rows  of a permutation matrix of order $m$, so that, without loss of generality,  $Y$ has the quantum Ito matrix $D\Omega D^\rT = I_r + i I_{r/2} \ox \bJ$ in view of (\ref{Omega}) and the Ito product rules (in particular, if $r=m$ and $D=I_m$, then $Y$ includes all the output fields).
The corresponding system-field space $\fH := \fH_0 \ox \fF$ is the tensor product of the initial system space $\fH_0$ (for the action of  $X_1(0), \ldots, X_n(0)$) and the Fock space $\fF$. Also, the system-field quantum state is assumed to be in the form
\begin{equation}
\label{rho}
    \rho:= \rho_0\ox \ups,
\end{equation}
where $\rho_0$ is the initial system state on $\fH_0$, and $\ups$ is the vacuum field state \cite{P_1992} on the Fock space $\fF$.
  While the internal energy of the OQHO is described by the Hamiltonian $H$ in  (\ref{H}) as before, the system-field coupling (due to the energy exchange between the oscillator and the external quantum fields) is quantified by a coupling matrix $N\in \mR^{m\x n}$. The latter parameterizes the linear dependence of a vector  $NX$ of $m$ self-adjoint coupling operators on the system variables. Accordingly, the matrices $A$, $B$, $C$ in (\ref{dXdY}) depend on the energy and coupling matrices $R$, $N$ and the CCR matrices  $\Theta$, $J$  from (\ref{XCCR}), (\ref{Omega}) as
\begin{align}
\label{AA}
    A & = A_0 +\wt{A},
    \qquad
    \wt{A}:= 2\Theta N^\rT J N
    =
    -\tfrac{1}{2}BJB^\rT \Theta^{-1},\\
\label{BC}
     B & = 2\Theta N^\rT,
     \qquad\
     C = 2DJN,
\end{align}
which imposes physical realizability (PR) constraints \cite{JNP_2008,SP_2012} on $A$, $B$, $C$.
If the oscillator is decoupled from the environment (that is, $N=0$), the matrices $\wt{A}$, $B$, $C$ vanish. In this case, $A=A_0$ as given by (\ref{A0}),  and the first QSDE in (\ref{dXdY}) reduces to the ODE (\ref{Xdot}), while the second QSDE in (\ref{dXdY}) takes the form $\rd Y = D\rd W$, so that the output field $Y$ is not affected by the system. In general, the matrix $A$ in (\ref{AA}) is decomposed into its Hamiltonian part $A_0$ and the skew Hamiltonian part $\wt{A}$ since the matrix $N^\rT J N$ is antisymmetric. Therefore, in the zero-Hamiltonian case of $R=0$, when $A_0 =0$,  the matrix $A = \wt{A}$ is skew Hamiltonian. Hence, by \cite[Theorem~3 on p.~145]{FMMX_1999}, there exists a Hamiltonian square root of $A$:
\begin{equation}
\label{root}
    A = (\Theta S)^2,
    \qquad
    S = S^\rT \in \mR^{n\x n},
\end{equation}
where, in view of the standing assumption (\ref{detTheta}),  the matrix $S$  satisfies
\begin{equation}
\label{STS}
    2N^\rT J N = S \Theta S.
\end{equation}
In this case, due to the squaring in (\ref{root}),  each nonzero eigenvalue of $A$ has an even multiplicity since the spectrum of the Hamiltonian matrix $\Theta S$ is centrally symmetric about the origin in the complex plane $\mC$. Furthermore, if the matrix $S$ is positive or negative definite, then, as discussed previously,  the spectrum of $\Theta S$ is purely imaginary, and the matrix $A$ is Hurwitz (since $(i\omega)^2 = -\omega^2 <0$ for any $\omega \in \mR\setminus \{0\}$). The latter property leads to an exponentially fast decay in the dependence of the solution
\begin{equation}
\label{Xsol}
    X(t)
    =
    \re^{tA} X(0)
    +
    Z(t),
    \qquad
    t\> 0,
\end{equation}
of the first QSDE in (\ref{dXdY}) on $X(0)$, as $t\to +\infty$. Here, in view of  (\ref{rho}),  with the input field $W$ being in the vacuum state $\ups$ on $\fF$,     the response
\begin{equation}
\label{Z}
    Z(t)
    :=
    \int_0^t
    \re^{(t-s)A}
    B
    \rd W(s)
\end{equation}
of the system variables to $W$ is a zero-mean Gaussian quantum process on the Fock space $\fF$ (see, for example, \cite[Section~3]{VPJ_2018a}), which commutes with and is statistically independent of $X(0)$. These properties of $Z$ hold regardless of whether the matrix $A$ is Hurwitz.

Now, the idea of using the zero-Hamiltonian OQHO for storage (and hence, avoiding the dissipation)   suggests that of particular interest are those coupling matrices $N$,  for which the matrix $S$ in (\ref{root}), (\ref{STS}) is, with necessity, not definite and is such that the spectrum of the Hamiltonian square root $\Theta S$ has a nonempty intersection with the set
\begin{equation}
\label{fS}
  \fS
  :=
  \{\lambda \in \mC: \Re (\lambda^2) = 0\}
  =
  \{\lambda \in \mC: |\Re \lambda| = |\Im \lambda|\}
\end{equation}
shown in Fig.~\ref{fig:fS}.
\begin{figure}[htbp]
\centering
\unitlength=1mm
\linethickness{0.5pt}
\begin{picture}(20.00,17.00)
    \put(0,8){\vector(1,0){20}}
    \put(10,-2){\vector(0,1){20}}
    \put(0,-2){\line(1,1){20}}
    \put(20,-2){\line(-1,1){20}}
    \put(19,5){\makebox(0,0)[cc]{\small$\Re$}}
    \put(7,17){\makebox(0,0)[cc]{\small$\Im$}}
    \put(8.9,5){\makebox(0,0)[cc]{\small$0$}}
\end{picture}
\caption{The set $\fS$ in (\ref{fS}) is the union of two straight lines bisecting the orthants of the complex plane. }
\label{fig:fS}
\end{figure}
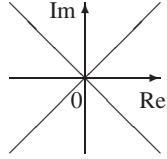
Indeed, for any such system-field coupling, at least one of the eigenvalues of the matrix $A$ in (\ref{root}) is imaginary. On the other hand, if the spectrum of the matrix $\Theta S$ is contained by the set $\fS$ in (\ref{fS}), so that all the eigenvalues of $A$ are either zero or purely imaginary and the dependence of $X(t)$ in (\ref{Xsol}) on $X(0)$ due to the term $\re^{tA}X(0)$ does not decay,  then the response $Z(t)$ to the input field in (\ref{Z}) grows with time $t$,  which makes this dependence increasingly distorted by the quantum noise. These two qualitatively different cases (when $A$ is Hurwitz and when its spectrum is purely imaginary) can be compared to each other in terms of the relative ``size'' of $\re^{tA}X(0)$ (as a useful ``signal'' carrying information about $X(0)$)  in comparison with the noise term $Z(t)$.

\section{DEVIATION FROM INITIAL CONDITION}
\label{sec:dev}

The  extent to which the dependence of $X(t)$ in (\ref{Xsol}) on $X(0)$ is noise-corrupted can be described in terms of the second-moment matrix
\begin{align}
\nonumber
    V(t)
     :=&
    \bE(Z(t)Z(t)^\rT)
    =\!\!
    \int_{[0,t]^2}
    \re^{(t-u)A}
    B
    \bE(\rd W(u)\rd W(v)^\rT)
    B^\rT
    \re^{(t-v)A^\rT}\\
\label{V}
    =&
    \int_0^t
    \re^{sA}
    B
    \Omega
    B^\rT
    \re^{sA^\rT}
    \rd s,
    \qquad
    t \> 0
\end{align}
(where $\Omega$ is the quantum Ito matrix from (\ref{Omega})),
which coincides with the finite-horizon controllability Gramian of the pair $(A, B\sqrt{\Omega})$ over the time interval $[0,t]$ and satisfies the Lyapunov ODE
$    \dot{V}(t) = AV(t) + V(t)A^\rT + B\Omega B^\rT
$,
with the zero initial condition $V(0) = 0$. 
The ``size'' of the  deviation
\begin{equation}
\label{xi}
  \xi(t)
  :=
  X(t)-X(0)
  =
  (\re^{tA}-I_n)X(0) + Z(t)
\end{equation}
of $X(t)$ from $X(0)$
can be quantified in a mean-square sense by using the matrix (\ref{V}) as
\begin{align}
\nonumber
    \Ups(t)
    := &
    \bE (\xi(t)\xi(t)^\rT)
    =
    (\re^{tA}-I_n)
    \Pi
    (\re^{tA^\rT}-I_n)
    +
    V(t)
    \\
\nonumber
    & +
    (\re^{tA}-I_n)
    \bE (X(0) Z(t)^\rT)
    +
    \bE (Z(t) X(0)^\rT)
        (\re^{tA^\rT}-I_n)\\
\label{Exixi}
        =&
    (\re^{tA}-I_n)
    \Pi
    (\re^{tA^\rT}-I_n)
    +
    V(t),
\end{align}
where $\bE (X(0) Z(t)^\rT) = \bE X(0) \bE Z(t)^\rT = 0$ due to the commutativity and statistical independence between $X(0)$ and $Z(t)$ combined with $\bE Z(t) = 0$. Here, $\Pi$ is the second-moment matrix of the initial system variables (assuming that $\bE(X(0)^\rT X(0)) = \sum_{k=1}^n \bE (X_k(0)^2)< +\infty$):
\begin{equation}
\label{EXX0}
    \Pi
    :=
    \bE (X(0)X(0)^\rT)
    =
    P + i\Theta,
    \qquad
    P:= \Re \Pi ,
\end{equation}
where the second equality follows from (\ref{XCCR}).
A scalar-valued mean-square measure for the deviation (\ref{xi}) is provided by
\begin{align}
\nonumber
    \Delta(t)
    & :=
    \bE(\xi(t)^\rT \Sigma\xi(t))
    =
    \bra \Sigma,  \Re \Ups(t)\ket\\
\label{ExiTxi}
    & =
    \|F(\re^{tA}-I_n)
    \sqrt{P}\| ^2 +
    \bra \Sigma, \Re V(t)\ket ,
\end{align}
where $0 \preccurlyeq\Sigma = \Sigma^\rT \in \mR^{n\x n}$ is a given weighting matrix which specifies the relative importance of the system variables, and  $\bra\cdot, \cdot\ket $ is   the Frobenius inner product generating the Frobenius norm $\|\cdot\| $ of matrices \cite{HJ_2007}. Here, use is made of the Hermitian property of the matrices $V(t)$, $\Ups(t)$, $\Pi$ in (\ref{V}), (\ref{Exixi}), (\ref{EXX0}),  whereby their imaginary parts  are antisymmetric matrices (in particular, $\Im \Pi = \Theta$) and hence, orthogonal to the subspace of symmetric matrices. The weighting matrix can be  factorized as
\begin{equation}
\label{FF}
    \Sigma := F^\rT F,
    \qquad
    F \in \mR^{s\x n},
    \qquad
    s:= \rank \Sigma \< n,
\end{equation}
so that (\ref{ExiTxi}) is concerned with $s$ independent linear combinations of the system variables of interest with real coefficients comprising the rows of the matrix $F$ (which is of full row rank).  
The Taylor series expansion of (\ref{ExiTxi}) as a function of time, truncated to the first two terms,  is  given by
\begin{equation}
\label{Delasy0}
    \Delta(t)
     =
     \dot{\Delta}
     t +
     \tfrac{1}{2}
     \ddot{\Delta}
     t^2
     +
     O(t^3),
     \qquad
     {\rm as}\
     t \to 0+,
\end{equation}
where the appropriate time derivatives of $\Delta$ at $t=0$  are computed as
\begin{equation}
\label{Deldot}
  \dot{\Delta}
  =
   \|F B\| ^2,
   \qquad
   \ddot{\Delta}
   =
   \bra \Sigma, ABB^\rT+BB^\rT A^\rT + 2A P A^\rT\ket .
\end{equation}
At the opposite extreme,
the infinite-horizon behaviour of (\ref{ExiTxi}) depends on whether the OQHO is dissipative. More precisely, if the matrix $A$ is Hurwitz, and hence, $\lim_{t\to +\infty}\re^{tA} = 0$,  then (\ref{V}), (\ref{ExiTxi}) imply that
$
    \lim_{t\to +\infty}
    \Delta(t)
    =
    \|F\sqrt{P}\| ^2 +
    \bra \Sigma, P_\infty\ket
    =
    \|F\sqrt{P+P_\infty}\|^2
$,
where
$
    P_\infty
    :=
    \int_0^{+\infty}
    \re^{tA}
    B
    B^\rT
    \re^{tA^\rT}
    \rd t
$
is the infinite-horizon controllability Gramian of the pair $(A,B)$ satisfying the algebraic Lyapunov equation (ALE) $AP_\infty + P_\infty A^\rT + BB^\rT = 0$. However, if the matrix $A$ has a purely imaginary spectrum $\{\pm i\omega_k: k =1, \ldots, \nu\}$ with pairwise different eigenfrequencies $\omega_k>0$ for $k = 1, \ldots, \nu$, then $\re^{tA}$ does not decay, as $t\to +\infty$. Indeed, in this case, $A$ is diagonalizable as
\begin{equation}
\label{AU}
    A = i U \mho U^{-1},
    \qquad
    \mho := {\diag}_{1\< k \< n} (\omega_k),
\end{equation}
where $U := [U_1\,   \ldots\,   U_n] \in \mC^{n\x n}$ is a nonsingular matrix whose columns $U_k:= (U_{jk})_{1\< j \< n} \in \mC^n$ are the corresponding eigenvectors  of $A$:
\begin{equation}
\label{AUk}
    A U_k = i\omega_k U_k,
    \qquad
    k = 1, \ldots, n.
\end{equation}
Here, without loss of generality and in accordance with (\ref{nnu}),
\begin{equation}
\label{omUk}
    \omega_{k+\nu}
    := -\omega_k,
    \qquad
    U_{k+\nu} = \overline{U_k},
    \qquad
    k = 1, \ldots, \nu.
\end{equation}
Therefore,
\begin{equation}
\label{etA}
    \re^{tA}
    =
    U\re^{it\mho}U^{-1}
    =
    \sum_{k = 1}^n
    U_k\re^{i\omega_k t} (U^{-1})_{k\bullet}
\end{equation}
(with $(\cdot)_{k\bullet}$ the $k$th row of a matrix) is an oscillatory function of time.   Hence, by combining the relation
$
    \lim_{t\to +\infty}
    \big(
    \frac{1}{t}
    \int_0^t
    \re^{i\omega s }\rd s
    \big)
    =
    \Big\{
    {\small\begin{matrix}
     1\ {\rm if}\    \omega = 0\\
     0\ {\rm if}\  \omega \ne 0
    \end{matrix}}
    $,  which holds for any $\omega \in \mR$,
with the assumption that the extended eigenfrequencies $\omega_1, \ldots, \omega_n$ in (\ref{AU})--(\ref{omUk}) are all different, it follows from (\ref{V}), (\ref{etA}) that
\begin{align}
\nonumber
    \tfrac{1}{t}
    V(t)
    & =
    \tfrac{1}{t}
    U
    \int_0^t
    \re^{is\mho}
    U^{-1}
    B\Omega B^\rT
    U^{-*}
    \re^{-is\mho}
    \rd s
    U^*\\
\nonumber
    & =
    \tfrac{1}{t}
    \sum_{j,k=1}^n
    \int_0^t
    \re^{i(\omega_j-\omega_k) s}
    \rd s\,
    U_j
    (U^{-1}
    B\Omega B^\rT
    U^{-*})_{jk}
    U_k^*
    \\
\nonumber
    & \to
    \sum_{k=1}^n
    U_k
    (U^{-1}
    B\Omega B^\rT
    U^{-*})_{kk}
    U_k^*\\
\label{Vasy}
    & =
    U
    (
    I_n
    \od
    (U^{-1}
    B\Omega B^\rT
    U^{-*}))
    U^*,
    \qquad
    {\rm as}\
    t \to +\infty,
\end{align}
where $A^\rT = A^* = -iU^{-*} \mho U^*$ since $A$ is real, with $(\cdot)^{-*}:= ((\cdot)^{-1})^*$, and
$\od$ is the Hadamard product \cite{HJ_2007}.  By (\ref{Vasy}), the matrix $V(t)$ grows asymptotically linearly in time $t$, and so does its real part in (\ref{ExiTxi}). The first term on the right-hand side of (\ref{ExiTxi}) is an oscillatory function with nonnegative values vanishing at $t=0$:
\begin{align*}
\nonumber
    \|&F(\re^{tA}-I_n)
    \sqrt{P}\| ^2
     =
    \|F U(\re^{it\mho}-I_n) U^{-1}
    \sqrt{P}\| ^2    \\
\nonumber
    = &
    \sum_{j,k=1}^n
    \bra
        F
        U_j (U^{-1})_{j\bullet}
        \sqrt{P},
        F
        U_k (U^{-1})_{k\bullet}
        \sqrt{P}
    \ket
    (\re^{-i\omega_jt}-1)
    (\re^{i\omega_kt}-1)
    \\
    = &
    \sum_{j,k=1}^n
    U_j^*\Sigma U_k
    (U^{-1})_{k\bullet}
    P
    (U^{-*})_j
    (\re^{i(\omega_k-\omega_j)t}-\re^{-i\omega_jt}-\re^{i\omega_kt}+1),
\end{align*}
where (\ref{etA}) is used again, with $(U^{-*})_j$ the $j$th column of $U^{-*}$.

\section{MEMORY DECOHERENCE TIME}
\label{sec:time}

The quality of the OQHO,  as a 
storage  device retaining memory about its initial system variables,   can be described by a typical ``decoherence'' time during which the system variables do not drift too far away from their initial values. In the weighted mean-square approach of (\ref{ExiTxi}),  such time can be defined as
\begin{equation}
\label{tau}
    \tau
    :=
    \inf
    \big\{t\> 0:\
    \Delta(t)
    >  \eps \|F\sqrt{P}\| ^2
    \big\},
\end{equation}
where the convention $\inf \emptyset := +\infty$ is used. Here, $\eps>0$ is a small dimensionless parameter specifying the relative error threshold for the deviation of $X(t)$ from $X(0)$,  
and it is assumed that
\begin{equation}
\label{pos}
    F
    \sqrt{P}
    \ne 0
\end{equation}
(recall that the columns of the matrix $F$ in (\ref{FF}) are linearly dependent if $s<n$).
Note that the memory decoherence time $\tau$, defined by (\ref{tau}),  depends on the fidelity parameter $\eps$ in a nondecreasing fashion. Furthermore, $\tau>0$ for any $\eps>0$ since $\Delta(t)$ is a continuous function of  $t\> 0$ satisfying $\Delta(0) = 0$.  Under a similar condition
\begin{equation}
\label{pos1}
    FB \ne 0,
\end{equation}
it follows from
(\ref{Delasy0}) that $\tau$ is asymptotically linear in $\eps$, with
\begin{equation}
\label{ratio}
    \lim_{\eps\to 0+}
    \frac{\tau}{\eps}
    =
    \frac{\|F\sqrt{P}\| ^2}{\dot{\Delta}}
    =
    \frac{\|F\sqrt{P}\| ^2}{\|FB\| ^2}
    =:
    \tau'
\end{equation}
being its derivative at $\eps =0$.
The quantity $\tau'$ is analogous to the signal-to-noise ratio since $P$ in the nominator pertains to the initial condition $X(0)$ to be stored (see (\ref{EXX0})), while $BB^\rT = \Re (B\Omega B^\rT)$ in the denominator is associated with the quantum diffusion matrix coming from the system-field  coupling in view of (\ref{Omega}), (\ref{BC}). However, $\tau'$  in (\ref{ratio}) has the physical dimension of time. By using both leading terms from (\ref{Delasy0}), the asymptotic relation (\ref{ratio}) can be extended to
\begin{equation}
\label{tau12}
    \tau =
      \wh{\tau}
     + O(\eps^3),
     \qquad
       \wh{\tau}
  :=
  \tau' \eps + \tfrac{1}{2}\tau'' \eps^2,
    \qquad
    {\rm as}\
    \eps \to 0+,
\end{equation}
where the second-order derivative $\tau''$  of $\tau$ at $\eps=0$ is obtained as
\begin{equation}
\label{tau2}
    \tau''
      =
    -
    \frac{\ddot{\Delta} \tau'^2}
    {\dot{\Delta}}
    =
    -
    \frac{\bra \Sigma, A BB^\rT + BB^\rT A^\rT  + 2APA^\rT\ket
    \|F\sqrt{P}\|^4}
    {\|FB\|^6}
\end{equation}
by matching the truncated expansions. Note that, in contrast to $\tau'$ in (\ref{ratio}), the coefficient $\tau''$ depends on both matrices $A$ and $B$, with the dependence on $A$ entering $\tau''$ only through $\ddot{\Delta}$ from (\ref{Deldot}).

A relevant performance criterion    for the OQHO as a quantum memory system   is provided by the maximization
\begin{equation}
\label{taumax}
  \tau\longrightarrow \sup
\end{equation}
of the decoherence time (or its second-order approximation $\wh{\tau}$
in (\ref{tau12})) at a given fidelity level $\eps$. Such maximization can be carried out by varying the energy and coupling parameters of the OQHO. 
Furthermore, the factor $F$ in (\ref{FF}) can also be part of the parameters over which $\tau$ is maximized. For example, the matrix $F$ can be varied so as to find a particular subset of the system variables of the OQHO (or their linear combinations) which are retained with the relative accuracy $\eps$  for a longer period of time than the others.

Note that the second-order approximation $\wh{\tau}$ of the decoherence time $\tau$ in (\ref{tau12}) depends on the matrix $A$ in a concave quadratic fashion, inheriting this property from $\tau''$ in (\ref{tau2}). Therefore, at least for small values of $\eps$,  the maximization (\ref{taumax}) tends to favour ``localized'' values  of $A$. More precisely, under the condition
\begin{equation}
\label{P0pos}
  P \succ 0
\end{equation}
on the matrix $P$ in (\ref{EXX0}), the completion of the square  leads to
$
    A BB^\rT + BB^\rT A^\rT  + 2APA^\rT
    =
    2 (A - \wh{A}) P (A - \wh{A})^\rT
    -
    \frac{1}{2}
    BB^\rT P^{-1} BB^\rT
$,
where
\begin{equation}
\label{Ahat}
    \wh{A}
    :=
    -\tfrac{1}{2}
    BB^\rT P^{-1} .
\end{equation}
Hence,  the $A$-dependent factor in the numerator of (\ref{tau2}) (see also (\ref{Deldot})) takes the form
\begin{equation}
\label{quad}
    \ddot{\Delta}
     =
    2
    \|F (A - \wh{A}) \sqrt{P}\|^2
    -
    \tfrac{1}{2}
    \|F BB^\rT P^{-1/2}\|^2
\end{equation}
and achieves $\min_{A \in \mR^{n\x n}} \ddot{\Delta} =
        - \frac{1}{2}
    \|F BB^\rT P^{-1/2}\|^2$
at $A = \wh{A}$.
The matrix $\wh{A}$ in (\ref{Ahat}) is isospectral to $P^{-1/2} \wh{A} \sqrt{P} = -\frac{1}{2} P^{-1/2} BB^\rT P^{-1/2} \preccurlyeq 0$, whereby all its eigenvalues are real and nonpositive.  The above remark on  the unconstrained minimum of $\ddot{\Delta}$ over $A$, based on (\ref{quad}),   ignores the PR structure (\ref{AA}) which makes $A$ quadratically dependent  on the coupling matrix $N$ and linearly dependent on the energy matrix $R$. The PR constraints are taken into account in the following theorem which provides a suboptimal solution of the problem (\ref{taumax}).

\begin{thm}
\label{th:R}
Suppose the fidelity level $\eps$ in (\ref{tau}) and the  weighting matrix $\Sigma$  in (\ref{FF}) are fixed together with the matrix $P$ in (\ref{EXX0}), and the conditions (\ref{detTheta}), (\ref{pos}), (\ref{pos1}) are satisfied.  Then for any given coupling matrix $N$ of the OQHO (\ref{dXdY}) in (\ref{AA}), (\ref{BC}), the energy matrix $R$ in (\ref{H}),  (\ref{A0})  delivers a solution to the problem
\begin{equation}
\label{tauhatmax}
  \wh{\tau} \longrightarrow \sup
  \quad{\rm over}\
  R = R^\rT \in \mR^{n\x n}
\end{equation}
of maximizing the approximate decoherence time $\wh{\tau}$ in
(\ref{tau12}) if and only if
\begin{equation}
\label{RASE}
    \Theta \Sigma \Theta R P  +
    P R \Theta \Sigma \Theta + K = 0,
\end{equation}
where $K = K^\rT \in \mR^{n\x n}$ is an auxiliary matrix  defined by
\begin{equation}
\label{K}
    K:=
    \tfrac{1}{4}
    (\Theta \Sigma (BB^\rT + 2\wt{A} P)
    -
     (BB^\rT + 2P \wt{A}^\rT) \Sigma \Theta ).
\end{equation}
\end{thm}

\begin{proof}
Since the coefficient $\tau'$ in (\ref{ratio}) does not depend on $R$, then $\sup_R \wh{\tau} = \tau' \eps + \frac{1}{2} \eps^2 \sup_R \tau''$ in view of (\ref{tau12}), and hence, (\ref{tauhatmax}) is equivalent to maximizing $\tau''$ in (\ref{tau2}) over $R$, which in turn reduces to
\begin{equation}
\label{Delmin}
  \ddot{\Delta}
  \longrightarrow \inf
  \quad{\rm over}\
  R = R^\rT \in \mR^{n\x n}.
\end{equation}
From the convex dependence of $\ddot{\Delta}$ on the matrix $A$ in (\ref{Deldot}) and the affine dependence of $A$ on $R$ in (\ref{AA}), it follows that (\ref{Delmin}) is a convex minimization problem on the subspace $\mS_n$ of real symmetric matrices in $ \mR^{n\x n}$. Hence, the first-order optimality condition
\begin{equation}
\label{opt}
    \d_R \ddot{\Delta} = 0
\end{equation}
for the Frechet derivative of $\ddot{\Delta}$ in $R$  is necessary and sufficient for $R$ to deliver a global  minimum to $\ddot{\Delta}$. By the second equality in (\ref{Deldot}), the first variation of $\ddot{\Delta}$ with respect to $A \in \mR^{n\x n}$ takes the form
\begin{align}
\nonumber
    \delta \ddot{\Delta}
    & =
    \bra
    \Sigma,
    (\delta A) BB^\rT+BB^\rT \delta A^\rT + 2(\delta A)  P A^\rT
    +
    2A  P \delta A^\rT
    \ket\\
\label{delta}
    & =
    2\bra
        \Sigma(BB^\rT + 2AP),
        \delta A
    \ket
\end{align}
due to the symmetry of $\Sigma$, $P$.
Since the matrix $\wt{A}$ in (\ref{AA}) is fixed, then  $\delta  A = \delta A_0 = 2\Theta \delta R$ in view of (\ref{A0}),  and its substitution into (\ref{delta}) yields
\begin{align}
\nonumber
    \delta \ddot{\Delta}
    & =
    4\bra
        \Sigma(BB^\rT + 2AP),
        \Theta \delta R
    \ket
    =
    -4\bra
        \Theta \Sigma(BB^\rT + 2AP),
        \delta R
    \ket    \\
\label{delta1}
    & =
    -4\bra
        \bS(\Theta \Sigma(BB^\rT + 2AP)),
        \delta R
    \ket.
\end{align}
Here, use is also made of the antisymmetry  of $\Theta$ and the symmetry of $R$ along with the matrix symmetrizer $\bS(L):= \frac{1}{2}(L+L^\rT)$ (so that $\bS$ is the orthogonal projection from $\mR^{n\x n}$ onto $\mS_n$). From (\ref{delta1}), it follows that $\d_R \ddot{\Delta} = -4\bS(\Theta \Sigma(BB^\rT + 2AP))$, whereby the condition (\ref{opt}), combined with (\ref{A0}), (\ref{AA}),   acquires the form
\begin{align}
\nonumber
    0 & =
    \bS(\Theta \Sigma(BB^\rT + 2AP))
    =
    \bS(\Theta \Sigma(BB^\rT + 2\wt{A}P))
    +
    4\bS(\Theta \Sigma \Theta R P)\\
\label{stat}
    & =
    2(K + \Theta \Sigma \Theta R P  +
    P R \Theta \Sigma \Theta),
\end{align}
which uses (\ref{K}) and establishes (\ref{RASE}) as a necessary and sufficient condition of optimality for the problem (\ref{tauhatmax}).
\end{proof}

By (\ref{RASE}), the zero energy matrix $R=0$ is an optimal solution of the problem (\ref{tauhatmax}) (and thus a suboptimal solution of (\ref{taumax}) for small values of $\eps$) if and only if the matrix $K$ in (\ref{K}) vanishes:
\begin{equation}
\label{K0}
    \Theta \Sigma (BB^\rT -BJB^\rT \Theta^{-1}P)
    -
     (BB^\rT - P \Theta^{-1}BJB^\rT ) \Sigma \Theta
     =
     0,
\end{equation}
where the second equality from (\ref{AA}) is used. Since $B$ in (\ref{BC}) depends linearly on the coupling matrix  $N$, the relation (\ref{K0}) is a quadratic constraint on $N$ under which $R=0$ is beneficial for maximizing the memory decoherence time of the OQHO in the framework of the approximation $\tau\approx \wh{\tau}$ in  (\ref{tau12}). Beyond the zero-Hamiltonian case, if (\ref{P0pos}) holds and the weighting matrix satisfies  $\Sigma \succ 0$, then $\Theta \Sigma \Theta = -\Theta \Sigma \Theta^\rT \prec  0$, thus making the matrix  $\Theta \Sigma \Theta P^{-1}$ Hurwitz and allowing $R$ to be found from (\ref{RASE}) in terms of an auxiliary matrix  $G:= P R P$  satisfying the following ALE:
\begin{equation}
\label{RG_GALE}
    R = P^{-1} G P^{-1},
    \qquad
    \Theta \Sigma \Theta P^{-1} G   +
    G P^{-1}\Theta \Sigma \Theta + K = 0.
\end{equation}
Note that the coupling matrix $N$ enters (\ref{RG_GALE}) through $K$ in (\ref{K}).

\section{COHERENT OQHO INTERCONNECTION}
\label{sec:two}

In addition to varying the energy and coupling parameters of a single OQHO for maximizing the memory decorence time $\tau$ in (\ref{tau}), similar settings can be formulated for interconnections of such systems. As an example, consider two OQHOs from \cite[Section~8]{VP_2023_SCL} which interact with external input bosonic fields and    are coupled to each other through a direct energy coupling and an indirect field-mediated coupling \cite{ZJ_2011a}; see
Fig.~\ref{fig:system}.
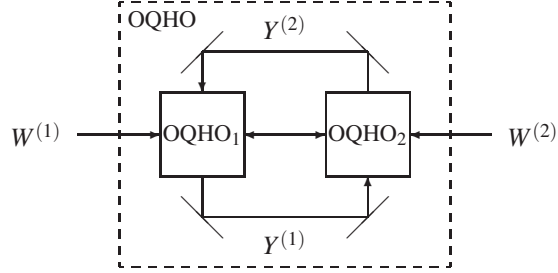
\begin{figure}[htbp]
\centering
\unitlength=1.1mm
\linethickness{0.5pt}
\begin{picture}(50.00,32.00)
    \put(6,29){\makebox(0,0)[lt]{\small OQHO}}
    \put(5,-2){\dashbox(40,32)[cc]{}}
    \put(10,9){\framebox(10,10)[cc]{\small OQHO${}_1$}}
    \put(30,9){\framebox(10,10)[cc]{\small OQHO${}_2$}}
    \put(0,14){\vector(1,0){10}}
    \put(50,14){\vector(-1,0){10}}
     \put(35,19){\line(0,1){5}}
    \put(35,24){\line(-1,0){20}}
    \put(15,24){\vector(0,-1){5}}
     \put(15,9){\line(0,-1){5}}
    \put(15,4){\line(1,0){20}}
    \put(35,4){\vector(0,1){5}}
    \put(32.5,1.5){\line(1,1){5}}
    \put(32.5,26.5){\line(1,-1){5}}
    \put(12.5,21.5){\line(1,1){5}}
    \put(12.5,6.5){\line(1,-1){5}}
    \put(20,14){\vector(1,0){10}}
    \put(30,14){\vector(-1,0){10}}
    \put(-2,14){\makebox(0,0)[rc]{$W^{(1)}$}}
    \put(25,27){\makebox(0,0)[cc]{$Y^{(2)}$}}
    \put(52,14){\makebox(0,0)[lc]{$W^{(2)}$}}
    \put(25,1){\makebox(0,0)[cc]{$Y^{(1)}$}}
\end{picture}
\caption{
    A coherent interconnection of two OQHOs, interacting with external input quantum Wiener processes  $W^{(1)}$, $W^{(2)}$  and coupled to each other through a direct energy coupling (depicted by a double arrow) and a field-mediated coupling through the quantum Ito processes $Y^{(1)}$, $Y^{(2)}$ at their corresponding outputs.
}
\label{fig:system}
\end{figure}
In this  coherent (measurement-free) feedback interconnection,
the external fields  are modelled by
quantum Wiener processes  $W^{(1)}$, $W^{(2)}$ (of even dimensions $m_1$, $m_2$) on symmetric Fock spaces $\fF_1$, $\fF_2$, respectively. These processes form an augmented quantum Wiener process
$    W
    :=
    {\small\begin{bmatrix}
      W^{(1)}\\
      W^{(2)}
    \end{bmatrix}}
$ of dimension $m:= m_1+m_2$
on the composite Fock space $\fF:= \fF_1\ox \fF_2$ with the quantum Ito matrix $\Omega$ in (\ref{Omega}) and the individual Ito tables $
    \rd W^{(k)}\rd W^{(k)}{}^{\rT} = \Omega_k \rd t$, where
$
    \Omega_k
     := I_{m_k} + iJ_k$ and
$
    J_k:= I_{m_k/2} \ox \bJ
$, 
with the matrix
$\bJ$ from (\ref{bJ}), so that
$    \Omega
    =
    {\small\begin{bmatrix}
      \Omega_1 & 0\\
      0 & \Omega_2
    \end{bmatrix}}$ and $
    J=
    {\small\begin{bmatrix}
      J_1 & 0\\
      0 & J_2
    \end{bmatrix}}$, whereby $W^{(1)}$, $W^{(2)}$ commute with (and are statistically independent of) each other. The constituent OQHOs are endowed with initial spaces $\fH_k$ and vectors $X^{(k)}$ of even numbers $n_k$ of dynamic variables on the composite system-field space $\fH:= \fH_0\ox \fF$, where $\fH_0:= \fH_1 \ox \fH_2$. Accordingly, the vector $    X
    :=
    {\small\begin{bmatrix}
      X^{(1)}\\
      X^{(2)}
    \end{bmatrix}}
$ of $n:= n_1+n_2$ system variables of the augmented OQHO   satisfies (\ref{XCCR}) with
\begin{equation}
\label{TTT}
    \Theta
    :=
    {\begin{bmatrix}
      \Theta_1 & 0\\
      0 & \Theta_2
    \end{bmatrix}}
\end{equation}
formed from the individual CCR matrices $\Theta_k = -\Theta_k^\rT \in \mR^{n_k\x n_k}$ (which are assumed to satisfy $\det \Theta_k\ne 0$):
\begin{equation}
\label{Theta12}
    [X^{(k)},
      X^{(k)}{}^\rT] = 2i\Theta_k,
      \qquad
    [X^{(3-k)},
      X^{(k)}{}^\rT] = 0,
      \qquad
      k = 1, 2.
\end{equation}
The direct energy  coupling of the OQHOs (see the double arrow in Fig.~\ref{fig:system}) is modelled by complementing the individual Hamiltonians $H_k:= \frac{1}{2} X^{(k)}{}^\rT R_k X^{(k)}$ of the OQHOs, specified by their energy matrices $R_k = R_k^\rT \in \mR^{n_k \x n_k}$,  with an additional term
\begin{equation}
\label{Hint}
    H_{12}
    :=
    X^{(1)}{}^\rT R_{12} X^{(2)} = X^{(2)}{}^\rT R_{21} X^{(1)},
\end{equation}
parameterized by $R_{12} = R_{21}^\rT \in \mR^{n_1\x n_2}$ and using the commutativity in (\ref{Theta12}). With the additional indirect coupling of the OQHOs in Fig.~\ref{fig:system}, mediated by their output fields $Y^{(1)}$, $Y^{(2)}$ of even dimensions $r_1$, $r_2$, the interconnection is
governed by
\begin{align}
\label{x}
    \rd X^{(k)}
    & =
    (A_k X^{(k)} + F_k X^{(3-k)}) \rd t  +  B_k \rd W^{(k)}  + E_k \rd Y^{(3-k)} ,\\
\label{y}
    \rd Y^{(k)}
    & =
    C_k X^{(k)} \rd t  +  D_k \rd W^{(k)} ,
    \qquad
    k = 1, 2.
\end{align}
The matrices
$    A_k\in \mR^{n_k\x n_k}$,
    $
    B_k\in \mR^{n_k\x m_k}$,
    $
    C_k\in \mR^{r_k\x n_k}$,
    $
    E_k\in \mR^{n_k\x r_{3-k}}$,
    $
    F_k\in \mR^{n_k\x n_{3-k}}
$
are parameterized as
\begin{align}
\label{Ak_Bk}
    A_k
     & =
    2\Theta_k(R_k + N_k^{\rT}J_k N_k + L_k^{\rT}\wt{J}_{3-k}L_k),
    \qquad
        B_k
     = 2\Theta_k N_k^{\rT},\\
\label{Ck_Ek_Fk}
    C_k  & =2D_kJ_k N_k,
    \qquad
    E_k  = 2\Theta_k L_k^{\rT},
    \qquad
    F_k  = 2\Theta_k R_{k,3-k},
\end{align}
and the matrices     $
    D_k\in \mR^{r_k\x m_k}$ consist of conjugate pairs of $r_k \< m_k$ rows of permutation matrices of orders $m_k$ with $D_k \Omega_k D_k^\rT = I_{r_k} + i\wt{J}_k$, where
$    \wt{J}_k:= D_kJ_kD_k^{\rT}  = I_{r_k/2}\ox \bJ$. 
Here,
$N_k\in \mR^{m_k \x n_k}$,  $L_k\in \mR^{r_{3-k} \x n_k}$ are the matrices of coupling of the $k$th  OQHO to its external input field $W^{(k)}$ and the output $Y^{(3-k)}$ of the other OQHO, respectively.  By combining  (\ref{x}), (\ref{y}),  the closed-loop OQHO in Fig.~\ref{fig:system} satisfies the first QSDE in (\ref{dXdY}) with the matrices
 \begin{equation}
\label{cAB}
    A
    =
    {\begin{bmatrix}
        A_1 & F_1+ E_1C_2\\
        F_2 +E_2C_1 & A_2
    \end{bmatrix}},
    \qquad
    B
    =
    {\begin{bmatrix}
        B_1 & E_1D_2\\
        E_2D_1 & B_2
    \end{bmatrix}}
\end{equation}
and  by (\ref{Ak_Bk}), (\ref{Ck_Ek_Fk}) has the following energy and coupling matrices: 
\begin{equation}
\label{Rclos_Nclos}
    R  =
    R_0
    +
    \wt{R},
    \quad
  R_0 :=
  {\begin{bmatrix}
    R_1 & R_{12}\\
    R_{21} & R_2
  \end{bmatrix}},
  \quad
    N   =
    {\begin{bmatrix}
      N_1 & D_1^{\rT}L_2 \\
      D_2^{\rT}L_1 & N_2
    \end{bmatrix}}.
\end{equation}
Here, $R_0$ 
is the closed-loop system  energy matrix in the absence of the indirect field-mediated coupling (when $L_1 = 0$, $L_2 = 0$, and the Hamiltonian reduces to $H_1+H_2 + H_{12}$ in view of (\ref{Hint})), while
\begin{equation}
\label{Rt}
    \wt{R}
    =
    {\begin{bmatrix}
      0                                       & L_1^{\rT}D_2J_2 N_2 -N_1^\rT J_1 D_1^\rT L_2\\
      L_2^{\rT}D_1J_1 N_1 -N_2^\rT J_2 D_2^\rT L_1  & 0
    \end{bmatrix}}
\end{equation}
comes from the field-mediated coupling between the OQHOs.
The closed-loop coupling matrix $N$ in (\ref{Rclos_Nclos}) depends linearly on the individual coupling matrices $N_1$, $L_1$, $N_2$, $L_2$, while $\wt{R}$ in (\ref{Rt}) depends  on them in a quadratic fashion. The energy matrix $R$ in (\ref{Rclos_Nclos}) can be assigned any given value in $\mS_n$ by an appropriate choice of $R_0$ in (\ref{Rclos_Nclos}). In particular, the zero Hamiltonian for the closed-loop OQHO (that is, with $R=0$) is achieved by letting
\begin{equation}
\label{RRR}
    R_1 = 0,
    \qquad
    R_2 = 0,
    \qquad
    R_{12}
    =
    N_1^\rT J_1 D_1^\rT L_2 - L_1^{\rT}D_2J_2 N_2,
\end{equation}
in which case  zero-Hamiltonian OQHOs are interconnected in such a way that the direct coupling Hamiltonian (\ref{Hint}) counterbalances the effect of the field-mediated coupling, as seen from the off-diagonal blocks of the matrices $R$, $R_0$, $\wt{R}$  in (\ref{Rclos_Nclos}), (\ref{Rt}). Instead of (\ref{RRR}), the direct energy coupling matrix $R_{12}$ can be found as a solution of the optimization problem
\begin{equation}
\label{tauhatmaxR12}
  \wh{\tau} \longrightarrow \sup
  \quad{\rm over}\
  R_{12}\in \mR^{n_1\x n_2},
\end{equation}
similar to (\ref{tauhatmax}) except that $R_1$, $R_2$ (in general, nonzero)  are fixed. 
\begin{thm}
\label{th:R12}
Suppose the OQHO interconnection described by (\ref{TTT})--(\ref{Rt}) satisfies the assumptions of Theorem~\ref{th:R}. Then the matrix $R_{12}$ in (\ref{Hint}) solves the problem (\ref{tauhatmaxR12}) with the approximate decoherence time $\wh{\tau}$ in
(\ref{tau12}) for the closed-loop OQHO if and only if
\begin{align}
\nonumber
    \Theta_1 &\Sigma_{11}\Theta_1 R_{12} P_{22}   +
    P_{11}R_{12} \Theta_2 \Sigma_{22}\Theta_2    \\
\label{RASE12}
    & + \Theta_1 \Sigma_{12}\Theta_2 R_{12}^\rT  P_{12}
    +
    P_{12}R_{12}^\rT \Theta_1 \Sigma_{12}\Theta_2
    +
Q = 0,
\end{align}
where $Q \in \mR^{n_1\x n_2}$ is an auxiliary matrix  which is computed as
\begin{equation}
\label{K12}
    Q\!:=\!
    \tfrac{1}{2}
    (\bS
    (\Theta \Sigma (BB^\rT\! +\! 2\breve{A}P))_{12},
    \
    \breve{A}
    \!:= \!
    2\Theta
    \Big(
{\begin{bmatrix}
  R_1 & 0\\
  0 & R_2
\end{bmatrix}}
    \!+\! \wt{R} \!+\! N^\rT JN
    \Big),
\end{equation}
in terms of (\ref{Rclos_Nclos}), (\ref{Rt}),
with $(\cdot)_{jk}$ the appropriate matrix blocks.
\end{thm}
\begin{proof}
The proof follows the lines of the proof of Theorem~\ref{th:R} except that the optimality condition (\ref{opt}) is replaced with
\begin{equation}
\label{opt12}
    \d_{R_{12}} \ddot{\Delta} = 0.
\end{equation}
Accordingly,  the energy matrix $R$ in (\ref{Rclos_Nclos}) is varied over an appropriate subspace of $\mS_n$ as $\delta R  = \delta R_0 = {\small\begin{bmatrix}
  0 & \delta R_{12}\\
  \delta R_{12}^\rT & 0
\end{bmatrix}}$ because the matrices $R_1$, $R_2$, $\wt{R}$ are fixed,    and hence, (\ref{delta1}) reduces to
$    \delta \ddot{\Delta}
       =
    -4\Big\bra
        \bS(\Theta \Sigma(BB^\rT + 2AP)),
        {\small\begin{bmatrix}
          0 & \delta R_{12}\\
        \delta R_{12}^\rT & 0
        \end{bmatrix}}
    \Big\ket
    =
    -8\bra
        (\bS(\Theta \Sigma(BB^\rT + 2AP)))_{12},
        \delta R_{12}
    \ket
$,
whereby $\d_{R_{12}} \ddot{\Delta} = -8(\bS(\Theta \Sigma(BB^\rT + 2AP)))_{12}$. Therefore, (\ref{opt12})  is  equivalent  to  an  appropriate   modification  of  (\ref{stat}):
$
    0  =
    \frac{1}{2}(\bS(\Theta \Sigma(BB^\rT + 2AP)))_{12}
    =
    \frac{1}{2}(\bS(\Theta \Sigma(BB^\rT + 2\breve{A}P)))_{12}
    +
    2\Big(\bS\Big(\Theta \Sigma \Theta         {\small\begin{bmatrix}
          0 & R_{12}\\
        R_{12}^\rT & 0
        \end{bmatrix}} P\Big)\Big)_{12}
     =
    Q +
    \Theta_1 \Sigma_{11}\Theta_1 R_{12} P_{22}  +
    \Theta_1 \Sigma_{12}\Theta_2 R_{12}^\rT  P_{12}
    +
    P_{12}R_{12}^\rT \Theta_1 \Sigma_{12}\Theta_2
    +
    P_{11}R_{12} \Theta_2 \Sigma_{22}\Theta_2
$,
which uses (\ref{K12}) and the representation $A= 2\Theta{\small\begin{bmatrix}
          0 & R_{12}\\
        R_{12}^\rT & 0
        \end{bmatrix}} + \breve{A}$ for the matrix $A$ in (\ref{cAB}) along with the block diagonal structure (\ref{TTT}) of the CCR matrix $\Theta$ and establishes (\ref{RASE12}) as a necessary and sufficient condition of optimality for the problem (\ref{tauhatmaxR12}).
\end{proof}

The relation (\ref{RASE12}) reduces to an algebraic Sylvester equation
$\Theta_1 \Sigma_{11}\Theta_1 R_{12} P_{22}   +
    P_{11}R_{12} \Theta_2 \Sigma_{22}\Theta_2 + Q=0$ for the matrix $R_{12}$
if the weighting matrix $\Sigma$ or the initial second-moment matrix $P$ for the system variables of the composite OQHO is block diagonal (that is, $\Sigma_{12}=\Sigma_{21}^\rT = 0$ or $P_{12} = P_{21}^\rT = 0$). Note that Theorem~\ref{th:R12} demonstrates one of possible approaches to optimizing quantum system interconnections for quantum memory applications.

\end{document}